\documentclass{article}

\usepackage{hyperref}
\usepackage{epigraph}
\usepackage{booktabs}
\usepackage{times}
\usepackage{helvet}
\usepackage{courier}
\usepackage{tabularx}
\usepackage{amsfonts}
\usepackage{graphicx}
\usepackage{caption}
\usepackage{subcaption}
\usepackage{lscape}
\usepackage{amsmath}
\usepackage{amssymb}
\usepackage{amsthm}
\usepackage{algorithm}
\usepackage{algorithmicx}
\usepackage{algpseudocode}
\usepackage[colorinlistoftodos]{todonotes}
\usepackage[shortlabels]{enumitem}
\usepackage{thm-restate}
\usepackage{epstopdf}
\usepackage{fullpage}

\newtheorem{theorem}{Theorem}

\newtheorem{example}{Example}
\newcommand{\leaveout}[1]{}

\algtext*{EndWhile}
\algtext*{EndIf}
\algtext*{EndForAll}
\algtext*{EndFor}
\algtext*{EndFunction}

\newcommand{\citep}[1]{\cite{#1}}
\newcommand{\citet}[1]{\cite{#1}}
\newcommand{\citeauthor}[1]{\cite{#1}}
\newcommand{\citeyear}[1]{\cite{#1}}

\title{A Rotating Proposer Mechanism for Team Formation}
\author{Jian Lou$^1$, Chen Hajaj$^2$, and Yevgeniy
  Vorobeychik$^4$\\
  $^1$Amazon, $^2$Ariel University, $^4$Washington University in St.~Louis
}
\date{}

\begin{document}

	





\maketitle

\section{The Rotating Proposer Mechanism}

A \emph{rotating proposer mechanism} is a way to operationalize
\emph{sequential proposer games}  in mechanism design for team formation.
formally define a \emph{team formation mechanism}.
A \emph{team formation
mechanism} $M$ maps every preference profile $\succ $ to a partition
$\pi $ of the set of players,
i.e. $\pi =M(\succ )$. Our goal is to exhibit such a mechanism, and analyze
its properties. The mechanism, termed Rotating Proposer Mechanism (RPM),
implements the subgame perfect Nash equilibrium of the sequential
proposer game with each proposer able to make an offer to each
possible team, in which all proposals are accepted (thus, only a
single offer is actually made).
In this equilibrium, whenever it's a player $i$'s
turn to propose, $i$ makes a proposal to her most preferred team among those
that would be accepted.

For any profile, if all players report their preferences truthfully,
equilibrium outcomes of the game have a number of good properties which are
thereby inhereted by RPM. Of particular note is that RPM is individually
rational, Pareto optimal, and implements iterated matching of
soulmates (IMS) (see \citep{leo2021matching}). However, it is also immediate
from known results that the RPM mechanism is not in general strategyproof
(this would conflict with individual rationality and implementing IMS~%
\citep{leo2021matching}).

The loss of incentive compatibility seems problematic.
However, one side-effect of RPM implementing IMS is that RPM is strongly
incentive compatible\footnote{%
More precisely, truth telling is a strong ex post Nash equilibrium.} and
yields a unique core team structure on a restricted class of preference
domains for which IMS always matches all players~\citep{leo2021matching}. As an
example, these domains include other well-known restrictions on preferences,
such as top coalition~\citep{banerjee_core_2001} and common ranking~%
\citep{Farrell88} properties.

This, however, would seem to limit its practical consideration, as such
restrictions can rarely be guaranteed or verified. Moreover, we wish to make
stronger efficiency claims than Pareto optimality, and also view fairness as
an important criterion. 
For the former, we are particularly interested in \emph{utilitarian social
welfare}, a much stronger criterion than Pareto efficiency. We will also
consider several notions of fairness discussed below.
While we cannot make strong theoretical guarantees about these for broad
realistic preference domains, we consider such properties empirically.


While RPM is a rather intuitive mechanism, it is quite challenging to
implement the associated subgame perfect Nash equilibrium. In particular,
the size of the backward induction search tree is $O(2^{\sum_{i=1}^{n} | 
\mathcal{T}_i|})$. Even in the roommate problem, in which the size of teams
is at most two, computing SPNE is $O({2^{n}}^2)$. We address this challenge
in three ways: (1) preprocessing and pruning to reduce the search space, (2)
approximation for the roommate problem, and (3) a general heuristic
implementation.

\subsection{Preprocessing and Pruning}

One of the central properties of RPM is that it implements iterative
matching of soulmates. In fact, it does so in every subgame in the backwards
induction process. Now, observe that computing the subset of teams produced
through IMS is $O(n^3)$ in general, and $O(n^2)$ for the roommate problem,
and is typically much faster in practice. We therefore use it as a
preprocessing step both initially (reducing the number of players we need to
consider in backwards induction) and in each subgame of the backwards
induction search tree (thereby pruning irrelevant subtrees).

Because IMS preprocessing is computationally efficient, it is always applied
before any of the approximate/heuristic versions of RPM below, with the
direct consequence that even these approximate versions implement IMS.

We show the computational value of IMS in preprocessing and pruning using
synthetic preference profiles based on the generative scale-free model. 

\begin{figure}[htb]
\begin{center}
\includegraphics[width = 80mm, 
height=43mm]{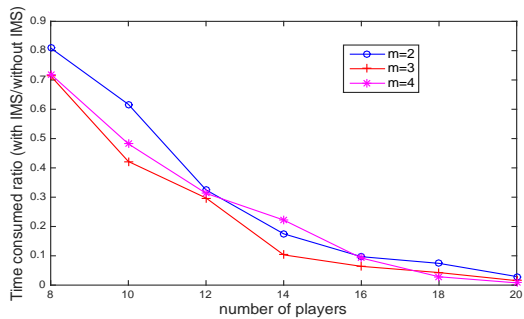}
\end{center}
\caption{Time consumed ratio (with IMS/without IMS) for RPM on scale-free
networks}
\label{fig:time_consumed}
\end{figure}


Figure \ref{fig:time_consumed} shows the ratio of time consumed by RPM with
IMS to that without IMS.\footnote{%
The simulations describes in this section were run on a 2.6 GHz Intel Core
i5 Mac machine with 8 GB RAM.} 
In all cases, we see a clear trend that using IMS in preprocessing and
pruning has increasing importance with increased problem size. 

\subsection{Approximate RPM for the Roommate Problem}

\label{subsec:alpha}

Using IMS for preprocessing and pruning does not sufficiently speed up RPM
computation in large-scale problem instances. Thus, we next developed a
parametric approximation of RPM that allows us to explicitly trade off
computational time against approximation quality. We leverage the
observation that the primary computational challenge of applying RPM to the
roommate problem is determining whether a proposal is to be accepted or
rejected. If we are to make this decision without exploring the full game
subtree associated with it, considerable time can be saved. Our approach is
to use a heuristic to evaluate the \textquotedblleft
likely\textquotedblright\ opportunity of getting a better teammate in later
stages: if this heuristic value is very low, the offer is accepted; if it is
very high, the offer is rejected; and we explore the full subgame in the
balance of instances.

More precisely, consider an arbitrary offer from player $i$ to another
player $j$. Given the subgame of the corresponding RPM, let $\mathcal{U}_{j}
(i)$ denote the set of feasible teammates that $j$ prefers to $i$, and let $%
\mathcal{U}_{j} (j)$ be the set of feasible teammates who $j $ prefers to be
alone. We can use these to heuristically compute the likelihood $R_j (i) $
that $j$ can find a better teammate than the proposer $i$: 
\begin{equation}
\begin{aligned} R_j (i)= & \frac{ |\mathcal{U}_{j} (i) |}{| \mathcal{U}_{j}
(j) |} \cdot \frac{1}{| \mathcal{U}_{j} (i) |} \sum_{k \in \mathcal{U}_{j}
(i)} \left( 1- \frac{ |\mathcal{U}_{k} (j) |}{ |\mathcal{U}_{k} (k)| }
\right) = & \frac{1}{ | \mathcal{U}_{j} (j) | } \sum_{k \in \mathcal{U}_{j}
(i)} \left( 1- \frac{ |\mathcal{U}_{k} (j) |}{ |\mathcal{U}_k (k)| } \right)
\end{aligned}  \label{eq:R_j}
\end{equation}
Intuitively, we first compute the proportion of feasible teammates that $j$
prefers to $i$. Then, for each such teammate $k$, we extract the proportion
of feasible teammates who are not more preferred by $k$ than the receiver $j$%
. Our heuristic then uses an exogenously specified threshold, $\alpha $, ($%
0\leq \alpha \leq 0.5$) as follows. If $R_j (i) \leq \alpha$, player $j$
accepts the proposal, while if $R_j (i) \geq 1-\alpha$, the proposal is
rejected. In the remaining cases, our heuristic proceeds with evaluating the
subgame at the associated decision node. Consequently, when $\alpha = 0$, it
is equivalent to the full backwards induction procedure, and computes the
exact RPM. Note that for any $\alpha$, this approximate RPM preserves IR,
and we also maintain IMS by running it as a preprocessing step.

The parameter $\alpha$ of our approximation method for RPM in the roommate
problem allows us to directly evaluate the trade-off between running time
and quality of approximation; small $\alpha$ will lead to less aggressive
use of the acceptance/rejection heuristic, with most evaluations involving
actual subgame search, while large $\alpha$ yields an increasingly heuristic
approach for computing RPM, with few subgames fully explored.

\begin{figure}[hbtp]
\centering
\begin{subfigure}[b]{0.45\textwidth}
		\includegraphics[width =\textwidth,  height=33mm]{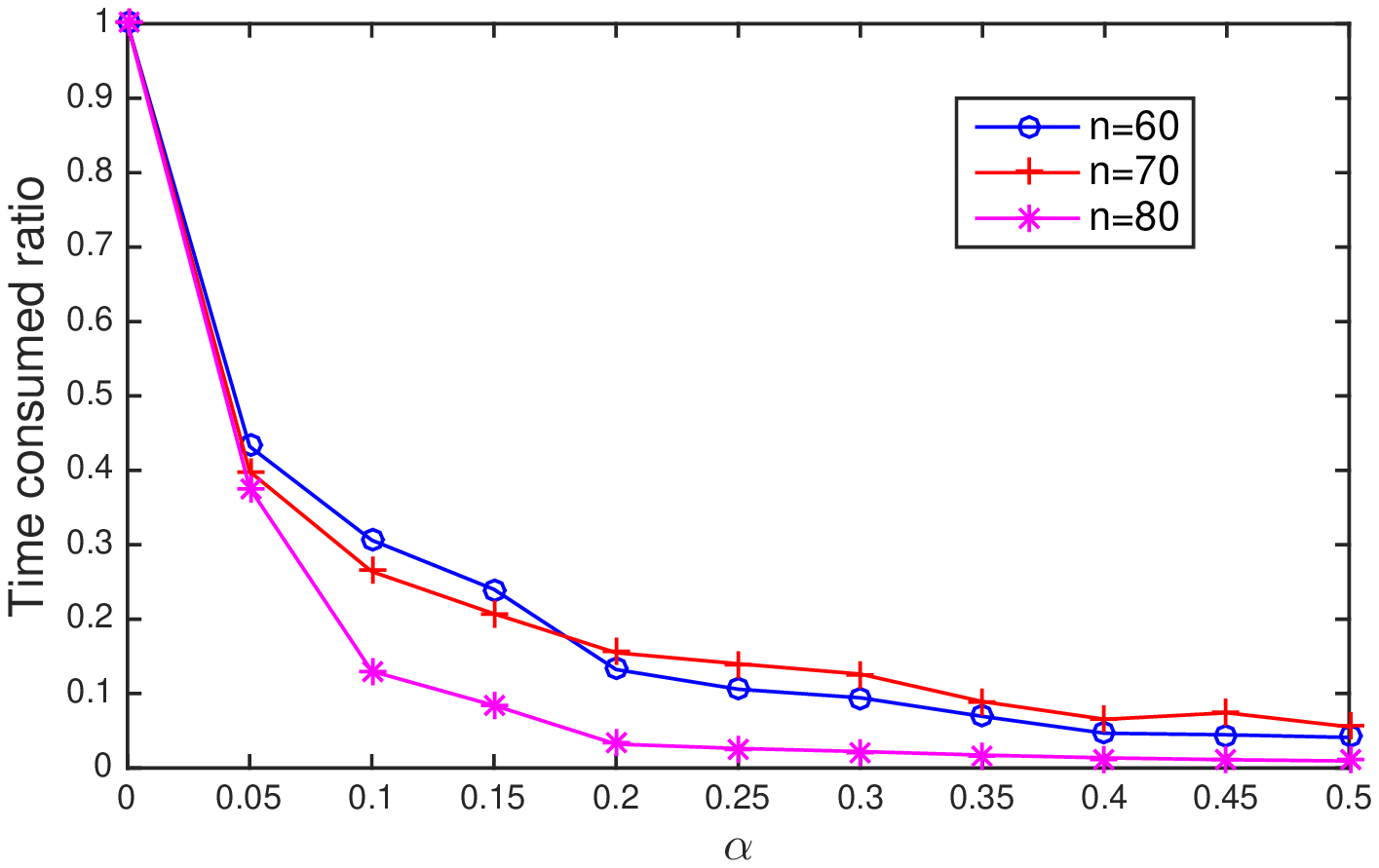}
		\caption{Time consumed ratio}
		\label{fig:Time_heuristic}
	\end{subfigure}
\begin{subfigure}[b]{0.45\textwidth}
		\includegraphics[width = \textwidth,  height=33mm]{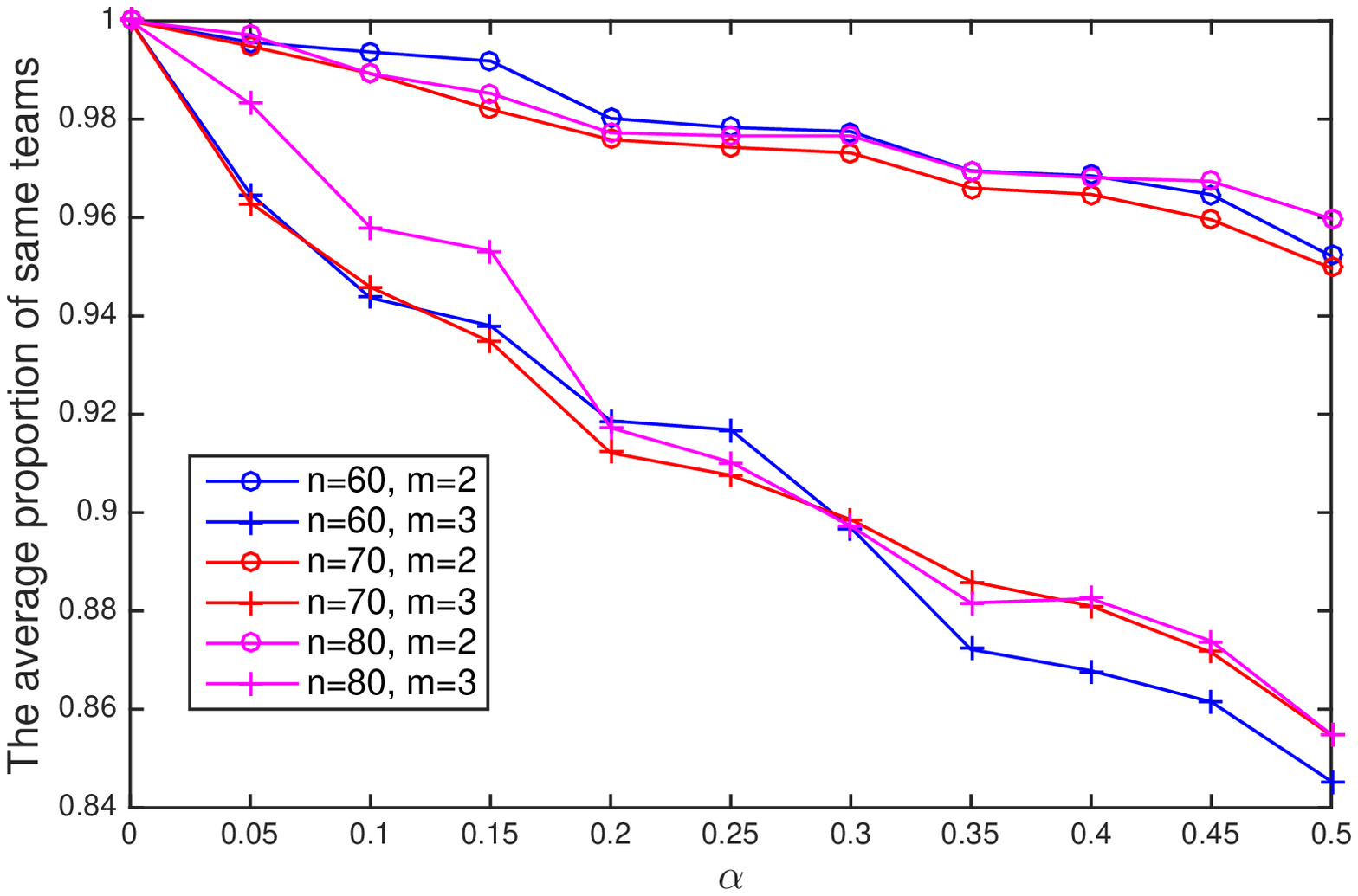}
		\caption{average proportion of same teams} 
		\label{fig:Team_heuristic}
	\end{subfigure}
\caption{Time consumed and average proportion of same teams}
\end{figure}

Figure \ref{fig:Time_heuristic} depicts the fraction of time consumed by RPM
with different values of $\alpha$ compared to exact RPM (when $\alpha = 0$)
on scale-free networks ($m=3$). Based on this figure, even a comparatively
small value of $\alpha$ dramatically decreases computation time.

Figure \ref{fig:Team_heuristic} compares similarity of the final team
partition when using the heuristic compared to the exact RPM. 
Notice that even for high values of $\alpha $, there is a significant
overlap between the outcomes selected by RPM with and without the heuristic. 
We note that $\alpha = 0.1$ appears to trade off approximation quality and
running time particularly well: for comparatively sparse networks (i.e., $%
m=2 $) it yields over $99\%$ overlap with exact RPM (this proportion is only
slightly worse for denser networks), at a small fraction of the running
time. Henceforth, we use $\alpha=0.1$ when referring to the approximate RPM
in the reminder of this section.

\subsection{Heuristic Rotating Proposer Mechanism (HRPM)}

\label{subsection:ARPM}

Unlike the roommate problem, general team formation problems have another
source of computational complexity: the need to iterate through the
combinatorial set of potential teams to propose to. Moreover, evaluating
acceptance and rejection becomes considerably more challenging. We therefore
develop a more general heuristic which scales far better than the approaches
above, but no longer has the exact RPM as a special case. We term the
resulting approximate mechanism \emph{Heuristic Rotating Proposer Mechanism
(HRPM)}, and it assumes that the sole constraint on teams is their
cardinality and that preferences can be represented by an additively
separable utility function~\citep{banerjee_core_2001}. With the latter
assumptions, we allow preferences over teams to be represented simply as
preference orders over potential teammates, avoiding the combinatorial
explosion in the size of the preference representation.

In HRPM, each proposer $i$ attempts to add a single member to their team at
a time in the order of preferences over players. If the potential teammate $%
j $ accepts $i$'s proposal, $j$ is added to $i$'s team, and $i$ proposes to
the next prospective teammate until either the team size constraint is
reached, or no one else who $i$ prefers to being alone is willing to join
the team. Player $j$'s decision to accept or reject $i$'s proposal is based
on calculating $R_j(l)$ for each member $l$ of $i$'s current team $T$ using
Equation~\ref{eq:R_j}, and then computing the average for the entire team, $%
R_j(T)= \frac{1}{|T|} \sum_{l\in T} R_j(l)$ (see Algorithm \ref{alg:ARPM} 
for the precise description of HRPM). We then use an exogenously specified
threshold $\beta \in [0,1]$, where $j$ accepts if $R_j(T)\leq \beta$ and
rejects otherwise. The advantage of HRPM is that the team partition can be
found in $O(\omega n^2)$, where $\omega$ is the maximum team size. The
disadvantage, of course, is that it only heuristically implements RPM.
Crucially, it does preserve IR, and IMS is implemented as a preprocessing
step. 
\begin{algorithm}[hbtp]
	\caption{Heuristic Rotating Proposer Mechanism (HRPM)}
	\label{alg:ARPM}
	\textbf{input:}  $(N , \succeq , O)$, $\omega$, $\beta$ \\
	\textbf{return:}  Team formation outcome $\pi$\\
	\begin{algorithmic}[1]
		
		\State $\pi= \emptyset$
		\While{$O$ is non-empty}
		\State $i \leftarrow$ the first player in $O$ 
		\State $\pi_i \leftarrow \{i\}$
		\While{$|\pi_i| < \omega$}
		\If{$\succeq_i$ is empty or the first player in $\succeq_i$ is $i$}
		\State $O\leftarrow O \backslash \{ i\}$
		\State break
		\EndIf
		\State player $i$ proposes to the first player $j$ in $\succeq_i$
		\State $R_j(\pi_i)= \frac{1}{|\pi_i|} \sum_{l\in \pi_i} \frac{1}{ | \mathcal{U}_{j} (j) | } \sum_{k \in \mathcal{U}_{j}
			(i)} \left( 1- \frac{ |\mathcal{U}_{k} (j) |}{ |\mathcal{U}_k (k)| } \right) $
		\If{$R_j(\pi_i) \leq \beta $}  \Comment{player $j$ accepts the proposal}
		\State $\pi_i \leftarrow \pi_i \cup \{ j\}$
		\State remove $ j $ from $ O $ and $ N $
		\EndIf
		\State remove $j$ from $\succeq_k$ for each player $k\in N$
		\EndWhile
		\State remove $ i $ from $ O $, $ N $ and  $\succeq_k$ for each player $k\in N$
		\EndWhile
		
		\While{$N$ is non-empty} \Comment{add singletons into the outcome.}
		\State pick an arbitrary instance $i$ from $N$
		\State remove $ i $ from $ O $ and $ N $
		\EndWhile
		\State return $\pi$
	\end{algorithmic}
      \end{algorithm}

\section{Properties of Exact and Approximate RPM}

Over truthful preference reports, RPM inherits the properties of the game,
including IR, IMS, and Pareto efficiency. In general, however, these
properties conflict with incentive compatibility. Moreover, when it comes to
efficiency, Pareto optimality is a weak criterion and we would wish to know
how well a mechanism fairs in terms of stronger efficiency criteria, such as
utilitarian social welfare (with cardinal preferences). Fairness, too, is an
important consideration in matching, particularly when it comes to forming
teams. Next, we explore these issues using empirical tools.

\subsection{Empirical Methodology}

In our empirical assessments, we use both synthetic and real hedonic
preference data. In both cases, preferences were generated based on a social
network structure in which a player $i$ is represented as a node and the
total order over neighbors is then generated randomly. Non-neighbors
represent undesirable teammates ($i$ would prefer being alone to being
teamed up with them).

The networks used for our experiments were generated using the following
models: 

\begin{itemize}
\item \textbf{Scale-free network:} We adapt the Barab\'{a}si-Albert model %
\citep{RevModPhys.74.47} to generate scale-free networks. For each $(n,m)$,
where $n$ is the number of players, $m$ denotes the density of the network,
we generate $1,000$ instances of networks and profiles.

\item \textbf{Karate-Club Network~\citep{Zac77}:} This network represents an
actual social network of friendships between 34 members of a karate club at
a US university, where links correspond to neighbors. We generate $100$
preference profiles based on the network.
\end{itemize}

Finally, we used a \emph{Newfrat} dataset~\citep{newfrat} that contains $15$
matrices recording weekly sociometric preference rankings from $17 $ men
attending the University of Michigan. In order to quantitatively evaluate
both the exact and approximate variants of RPM, the ordinal preferences $%
\succ_i$ have to be converted to cardinal ones $u_i(\cdot)$, upon which both
mechanisms operate. For this purpose, we introduce a \emph{scoring function}
suggested by~\citet{Bouveret:2011} to measure a player's utility. To compute
a player $i$'s utility of player $j$ we adopt \emph{normalized Borda scoring
function}, defined as $u_i(j) = g(r)=2(k-r+1)/k-1$, where $k$ is the number
of $i$'s neighbors, and $r \in \{1,\ldots,k\}$ is the rank of $j$ in $i$'s
preference list. Without loss of generality, for every player $i$ we set the
utility of being a singleton $u_i(i) = 0 $. We assume that the preferences
of players are additively separable~\citep{banerjee_core_2001}, which means
that a player $i$'s utility of a team $T$ is $u_i(T)= \sum_{j\in T} u_i(j)$.

\subsection{Incentive Compatibility}

\label{S:IC}

In spite of the known impossibility results, the fact that RPM is not
incentive compatible may be intuitively surprising, given that it implements
an equilibrium of the complete information game. 
To gain further intuition into this, consider the following example.

\begin{example}
Consider a roommate problem with 3 players having the following preferences: 
\begin{equation*}
\begin{matrix}
& 1 : & \{1, 2 \} & \succ_{1} & \{1, 3\} & \succ_{1} & \{ 1\} \\ 
& 2 : & \{2, 3 \} & \succ_{2} & \{2, 1\} & \succ_{2} & \{ 1\} \\ 
& 3 : & \{3, 1 \} & \succ_{3} & \{3, 2\} & \succ_{3} & \{ 1\} \\ 
&  &  &  &  &  & 
\end{matrix}%
\end{equation*}

Suppose that the order in RPM is $O=(1, 1, 1, 1, 2, 2, 2, 2, 3, 3, 3, 3)$.
In the subgame perfect Nash equilibrium of the corresponding RPG, 1 will
propose to $\{1, 2\}$, $2$ will accept, and the resulting teams are $\{
\{1,2\}, \{3\} \}$. This is because 2 is 1's most preferred roommate, and if
2 rejects, then 1 would offer to 3 who would accept (since they like 1 more
than 2), and 2 would be left alone.

Now, if player 3 misreports preferences to claim that she prefers 2 to 1,
then 2 and 3 are soulmates and would be matched, with the resulting outcome $%
\{\{1\}, \{2,3\}\}$. The latter outcome is clearly preferred by 3, and
consequently 3 has the incentive to lie.
\end{example}


Despite the general failure of incentive compatibility in RPM, we now
explore \emph{empirically} how frequently this failure actually occurs. We
use the roommate problem, as in this case the special structure of RPM
allows us to use Algorithm~\ref{alg:upperBound} to compute an upper bound on
the number of players who could possibly benefit by misreporting
preferences. 
In applying the algorithm, we use $\mathcal{T}_{i}$ to denote the set of
feasible teammates (since teams are of size at most 2).

At the high level, this algorithm considers all the players who have
accepted or rejected a proposal and checks whether reversing this decision
improves their outcomes. The following theorem shows that this method indeed
finds the upper bound on the number of untruthful players.

\begin{theorem}
Algorithm \ref{alg:upperBound} returns an upper bound on the number of
players who can gain by misreporting their preferences.
\end{theorem}

\begin{proof}
We divide the players into \emph{proposers} and \emph{%
receivers}. Proposers are those who propose in RPM and were thus teamed up
(including singleton teams). Receivers accept or reject someone's offer.

\begin{algorithm}[htbp]
	\caption{Computing Upper Bound of Untruthful Players}
	\label{alg:upperBound}
	\textbf{input:}  $(N, \succ, \mathcal{T}, O)$ , teammate vector $teammate[]$ which results from RPM\\
	\textbf{return:}  number of potential untruthful players $Sum$\\
	\begin{algorithmic}[1]		
		\State $sum\leftarrow0$
		\While{$|O| \geq 2$}
		\State $proposer$ $\leftarrow$ the first player in $O$
		\State $receiver$ $\leftarrow$ $teammate[proposer]$
		\For{player $i\in \mathcal{T}_{proposer}$}
		\If{ $i\succ_{proposer} receiver$  and $proposer \succ_{i} teammate[i]$ } 
		\State $sum\leftarrow sum+1$   \Comment{$i$ is potentially untruthful}
		\EndIf
		\EndFor
		\For{player $j\in \mathcal{T}_{receiver}$}
		\If{ $j\succ_{receiver} proposer$  and $receiver \succ_{j} teammate[j]$ } 
		\State $sum\leftarrow sum+1$   \Comment{$receiver$ is potentially untruthful}
		\EndIf
		\EndFor
		\State remove $proposer$ and $receiver$ from $N$, $O$ and $\mathcal{T}$
		\EndWhile
		\State return $sum$
	\end{algorithmic}
\end{algorithm}

There are 4 possible cases:

\begin{enumerate}
\item \emph{A proposer $i$ untruthfully reveals her preference and remains a
proposer.} 
As RPM implements subgame perfect Nash equilibrium in the corresponding
subgame, the proposer $i$ can match with the best roommate among those
accept her proposals by acting truthfully. Consequently, $i$ cannot improve
by lying.

\item \emph{A receiver $j$ untruthfully reveals her preference and is still
a receiver.} In this case, if $j$ has an incentive to lie, there has to be a
proposer $i^{\prime }$ who prefers $j$ to her teammate under RPM, while $j $
must prefer $i^{\prime }$ to her teammate. Steps $4-7$ in Algorithm \ref%
{alg:upperBound} count all such instances.

\item \emph{A proposer $i$ untruthfully reveals her preference and becomes a
receiver.} In this case, if $i$ has an incentive to untruthfully reveal her
preference, there has to be a proposer $i^{\prime }$ who prefers $i$ to
their teammate under RPM, and who $i$ also prefers to her teammate. Steps $%
4-7$ in Algorithm \ref{alg:upperBound} count all such instances.

\item \emph{A receiver $j$ untruthfully reveals her preference and becomes a
proposer.} In this case, if $j$ has an incentive to misreport her
preference, there must be a receiver $j^{\prime }$ who prefers $j$ to her
teammate, while $j$ must prefer $j^{\prime }$ to her teammate. Steps $8-10$
in Algorithm \ref{alg:upperBound} count all such instances.
\end{enumerate}
\end{proof}
This upper bound obtains for both the exact and approximate versions of RPM,
including HRPM. Next we evaluate the incentives to misreport preferences
using our RPM approximations in the context of the roommate problem.

\begin{table}[hbtp]
\caption{Average upper bound of untruthful players for (Approximate) RPM}
\label{Table:players}\setlength\tabcolsep{3pt}
\par
\begin{center}
{\small \ 
\begin{tabular}{|c|c|c|c|c|c|c|c|}
\hline
n & 20 & 30 & 40 & 50 & 60 & 70 & 80 \\ \hline
$m=2$, $\alpha=0$ & 0.015\% & 0.013\% & 0.013\% & 0.002\% & 0.008\% & 0.011\%
& 0.010\% \\ \hline
$m=2$, $\alpha=0.1$ & 0.015\% & 0.010\% & 0.015\% & 0.004\% & 0.022\% & 
0.029\% & 0.036\% \\ \hline
$m=3$, $\alpha=0$ & 0.105\% & 0.107\% & 0.072\% & 0.038\% & 0.037\% & 0.024\%
& 0.023\% \\ \hline
$m=3$, $\alpha=0.1$ & 0.115\% & 0.103\% & 0.085\% & 0.076\% & 0.065\% & 
0.074\% & 0.093\% \\ \hline
\end{tabular}
}
\end{center}
\end{table}

\begin{table}[hbtp]
\caption{Lower bound of profiles where every player is truthful for
(Approximate) RPM}
\label{Table:profiles}\setlength\tabcolsep{3pt}
\par
\begin{center}
{\small \ 
\begin{tabular}{|c|c|c|c|c|c|c|c|}
\hline
n & 20 & 30 & 40 & 50 & 60 & 70 & 80 \\ \hline
$m=2$, $\alpha=0$ & 99.7\% & 99.6\% & 99.5\% & 99.9\% & 99.6\% & 99.2\% & 
99.2\% \\ \hline
$m=2$, $\alpha=0.1$ & 99.7\% & 99.7\% & 99.4\% & 99.8\% & 98.8\% & 98.1\% & 
97.2\% \\ \hline
$m=3$, $\alpha=0$ & 97.9\% & 96.8\% & 97.1\% & 98.1\% & 97.8\% & 98.4\% & 
98.3\% \\ \hline
$m=3$, $\alpha=0.1$ & 97.8\% & 96.9\% & 96.8\% & 96.2\% & 96.3\% & 95.1\% & 
92.9\% \\ \hline
\end{tabular}
}
\end{center}
\end{table}

Table \ref{Table:players} presents the upper bound on the number of players
with an incentive to lie, as a proportion of all players, on scale-free
networks. We observe that the upper bound is always below $0.2\%$, and is
even lower when the networks are sparse ($m=2$). On the Karate club data, we
did not find any player with an incentive to lie in test cases when we apply
(Approximate) RPM. On the Newfrat data, the upper bounds are less than $7\%$
and $0.4\%$ when we apply RPM with and without heuristics, respectively. In
addition, we also computed the lower bound on the fraction of preference
profiles where truth telling is a Nash equilibrium (Table~\ref%
{Table:profiles}). We find that without the heuristic, when $m=2$ (sparse
networks), RPM is incentive compatible in more than $99\%$ of the profiles;
and when $m=3$ (the networks are comparatively dense), RPM is truthful at
least $96\%$ of the time.

\begin{table}[hbtp]
\caption{Average upper bound of untruthful players for HRPM}
\label{Table:ARPM}\setlength\tabcolsep{3pt}
\par
\begin{center}
{\small \ 
\begin{tabular}{|c|c|c|c|c|c|c|c|}
\hline
n & 20 & 30 & 40 & 50 & 60 & 70 & 80 \\ \hline
$m=2$, $\beta=0.5$ & 1.44\% & 1.77\% & 1.71\% & 2.00\% & 2.09\% & 2.16\% & 
2.06\% \\ \hline
$m=2$, $\beta=0.6$ & 1.62\% & 1.83\% & 1.96\% & 2.09\% & 2.25\% & 2.11\% & 
2.11\% \\ \hline
$m=3$, $\beta=0.5$ & 2.99\% & 3.36\% & 3.76\% & 3.90\% & 4.18\% & 4.02\% & 
4.33\% \\ \hline
$m=3$, $\beta=0.6$ & 3.44\% & 3.69\% & 3.97\% & 3.98\% & 4.40\% & 4.24\% & 
4.52\% \\ \hline
\end{tabular}
}
\end{center}
\end{table}

Table~\ref{Table:ARPM} presents the upper bound on the number of untruthful
players for HRPM (still for the roommate problem). Even with this heuristic,
we can see that fewer than 5\% of the players have any incentive to
misreport preferences in all cases.

\subsection{Efficiency}

\label{subsection:USW}

In terms of social welfare, ex post Pareto optimality, satisfied by both
random serial dictatorship (RSD)~\citep{Bade15} and RPM, is a very weak
criterion. Moreover, it is not necessarily satisfied by our approximations
of RPM. Conversion of ordinal to cardinal preferences allows us to
empirically consider \emph{utilitarian social welfare}, a much stronger
criterion commonly used in mechanism design with cardinal preferences.

We define utilitarian social welfare as $\frac{1}{|N|}\sum_{i\in N}
u_i(\pi_i) $, 
where $\pi_i$ is the team that $i$ was assigned to by the mechanism.

\begin{figure}[hbtp]
\centering
\begin{subfigure}[b]{0.48\textwidth}
		\includegraphics[width=\textwidth, height=35mm]{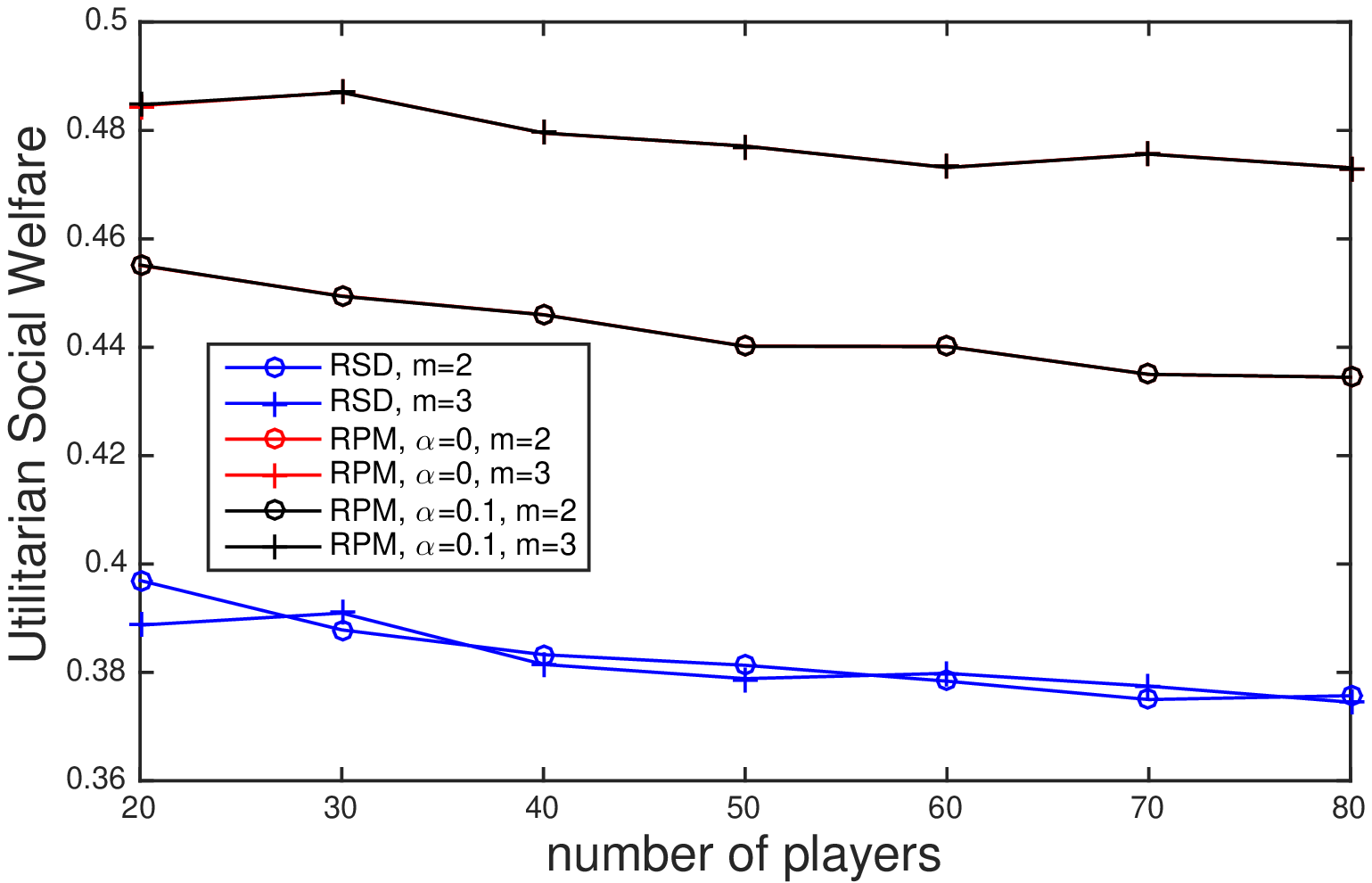} 
		\caption{scale-free networks}
		\label{fig:USW_Scale}
	\end{subfigure}
\begin{subfigure}[b]{0.48\textwidth}
		\includegraphics[width=\textwidth, height=35mm]{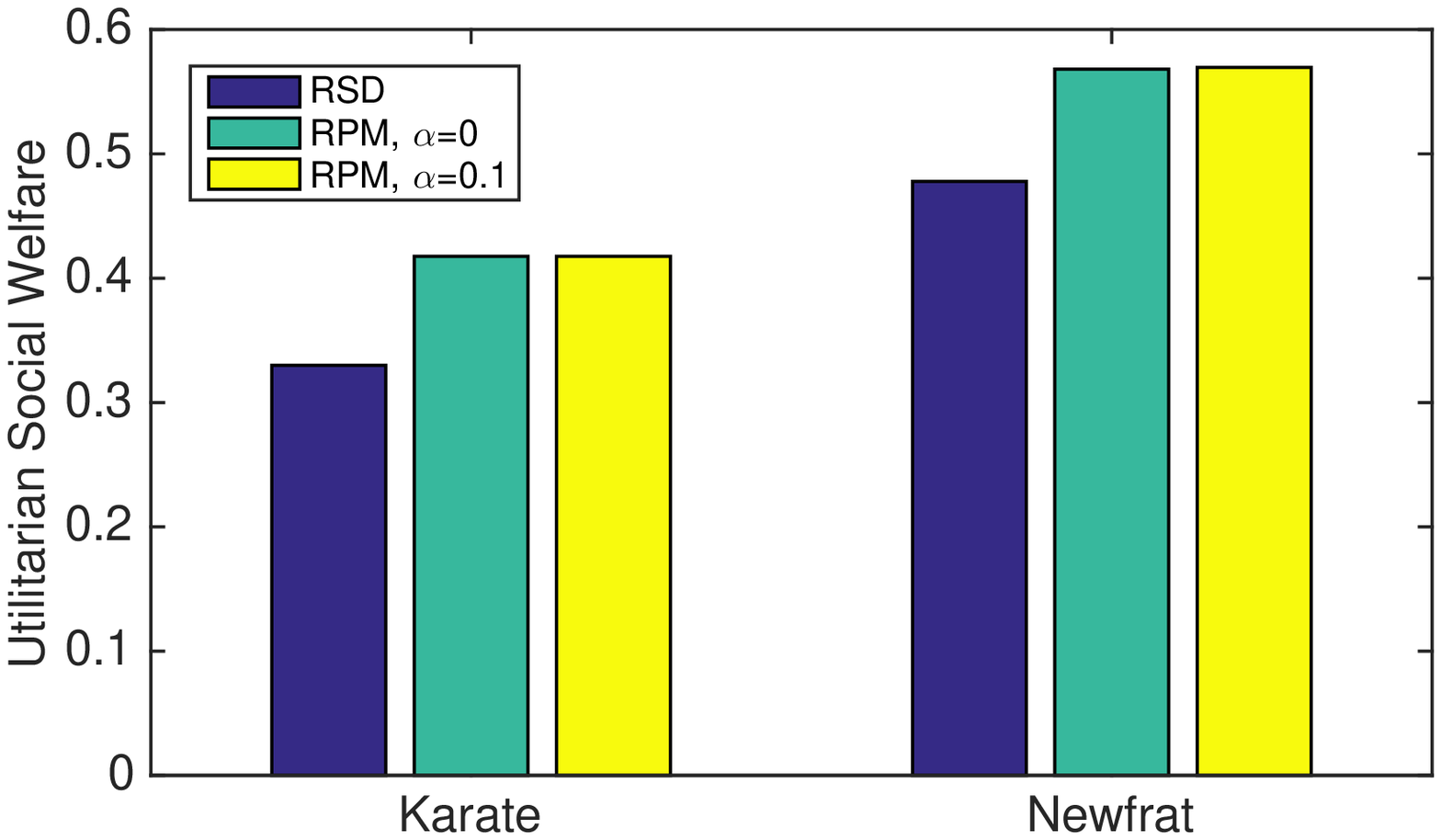} 
		\caption{Karate and Newfrat}
		\label{fig:USW_Real}
	\end{subfigure}
\caption{Utilitarian social welfare for roommate problem}
\end{figure}

Figures \ref{fig:USW_Scale} and \ref{fig:USW_Real} depict the average
utilitarian social welfare for RSD and RPM in the roommate problem on
scale-free networks, Karate club networks, and the Newfrat data. In all
cases, RPM yields significantly higher social welfare than RSD, with $%
15\%-20\% $ improvement in most cases. These results are statistically
significant ($p<0.01$). Furthermore, there is virtually no difference
between exact and approximate RPM.

\begin{figure}[hbtp]
\centering
\begin{subfigure}[b]{0.48\textwidth}
		\includegraphics[width=\textwidth, height=35mm]{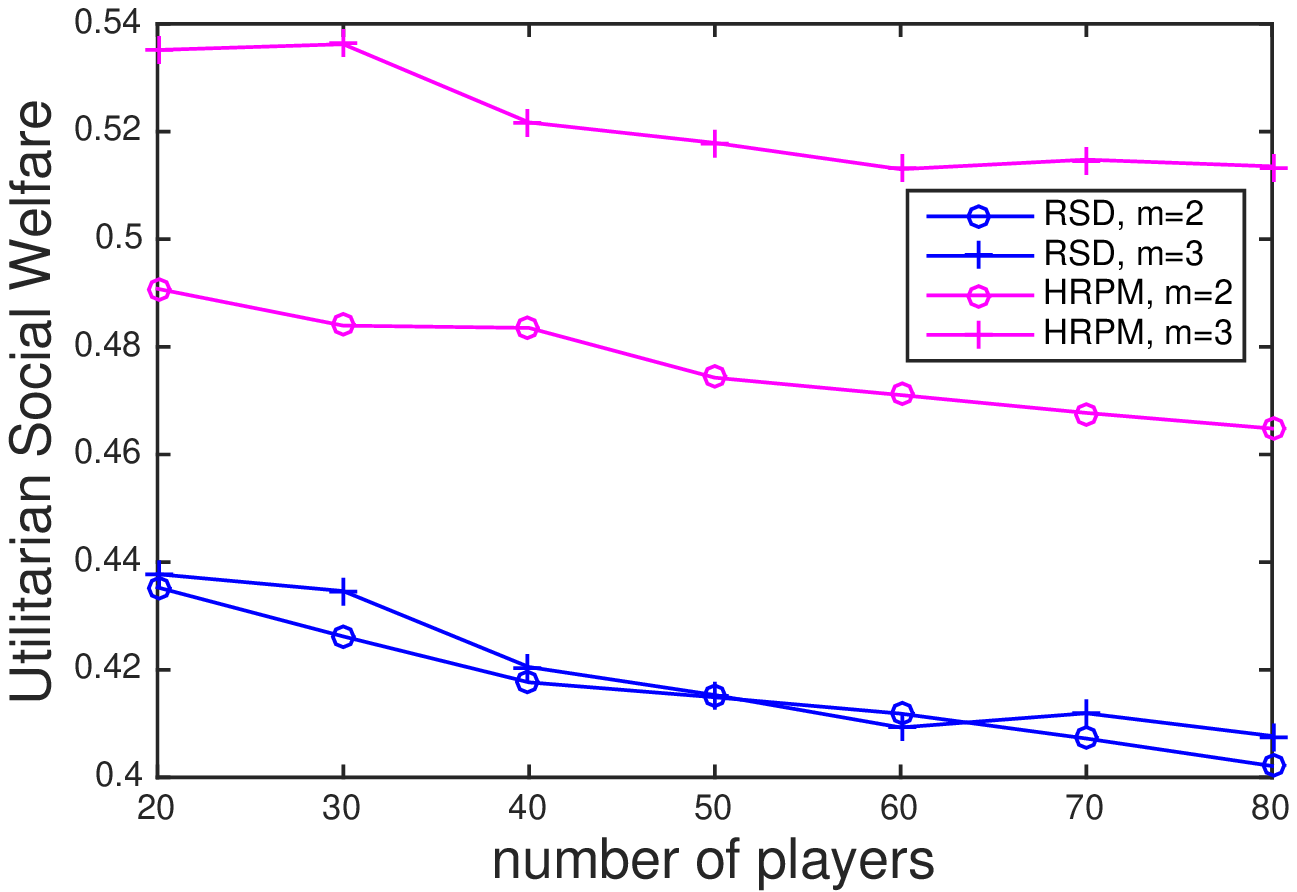} 
		\caption{scale-free networks}
		\label{fig:USW_Scale_three}
	\end{subfigure}
\begin{subfigure}[b]{0.48\textwidth}
		\includegraphics[width=\textwidth, height=35mm]{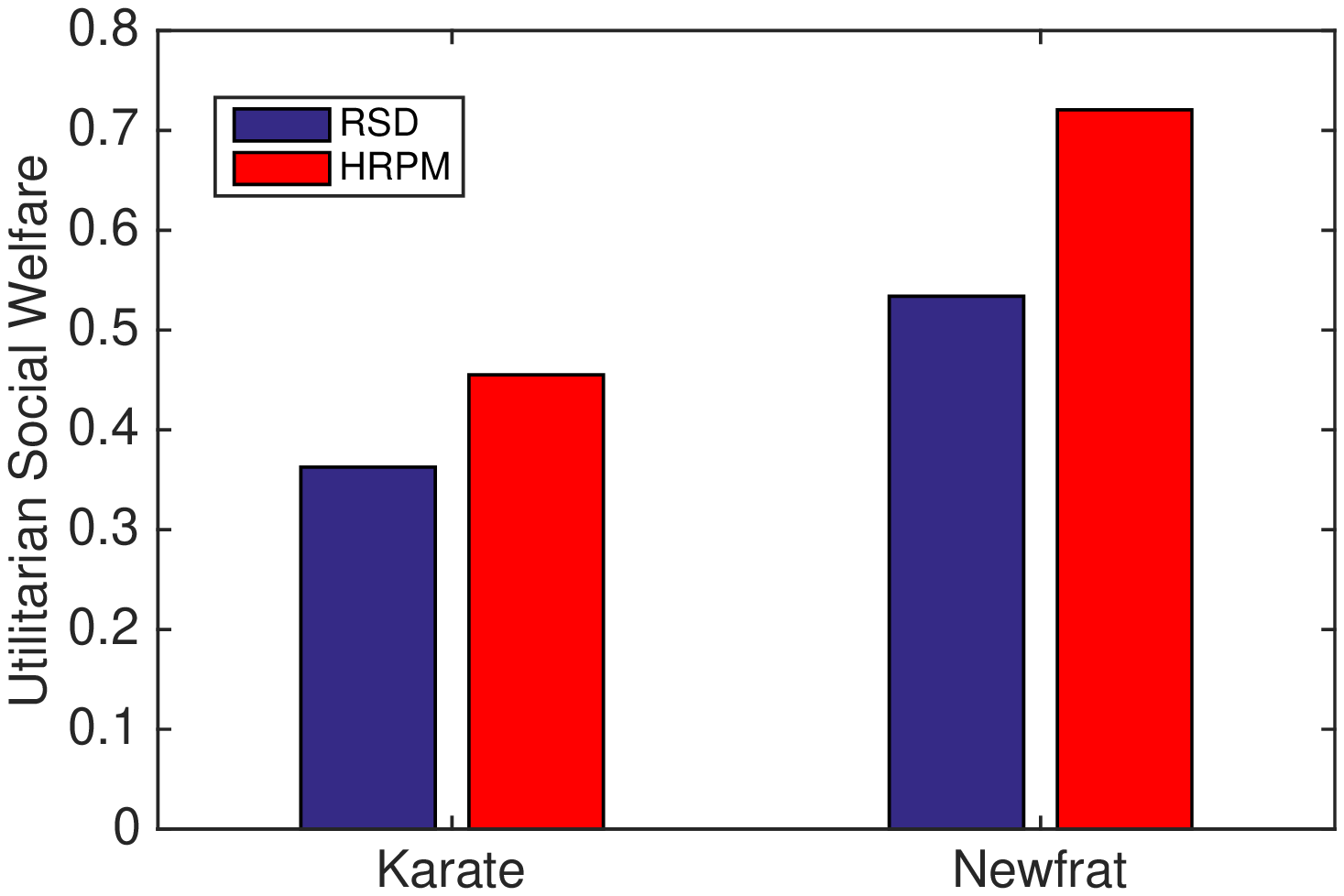} 
		\caption{Karate and Newfrat}
		\label{fig:USW_Real_three}
	\end{subfigure}
\caption{Utilitarian social welfare for trio-roommate problem}
\end{figure}

For the trio-roommate problem (in which the maximum size of team is 3), we
compare HRPM ($\beta=0.6$) with RSD on the same data sets. Figures \ref%
{fig:USW_Scale_three} and \ref{fig:USW_Real_three} show that HRPM yields
significantly higher social welfare than RSD in all instances, and HPRM
performs even better when the network is comparatively dense ($m=3$ in the
scale-free network). All results are statistically significant ($p<0.01$).

\subsection{Fairness}

\label{subsection:fairness}

A number of measures of fairness exist in prior literature. One common
measure, envy-freeness, is too weak to use, especially for the roommates
problem: every player who is not matched with her most preferred other will
envy someone else. Indeed, because RPM matches soulmates---in contrast to
RSD, which does not---it already guarantees the fewest number of envious
players in the roommates problem. 
We consider two alternative measures that aim to capture different and
complementary aspects of fairness: 
the Gini coefficient, representing the inequality among values of player
utilities, and the correlation between utility and rank in the random
proposer order (i.e., Pearson correlation).

\begin{figure}[hbtp]
\centering
\begin{subfigure}[b]{0.45\textwidth}
		\includegraphics[width = \textwidth, height=35mm]{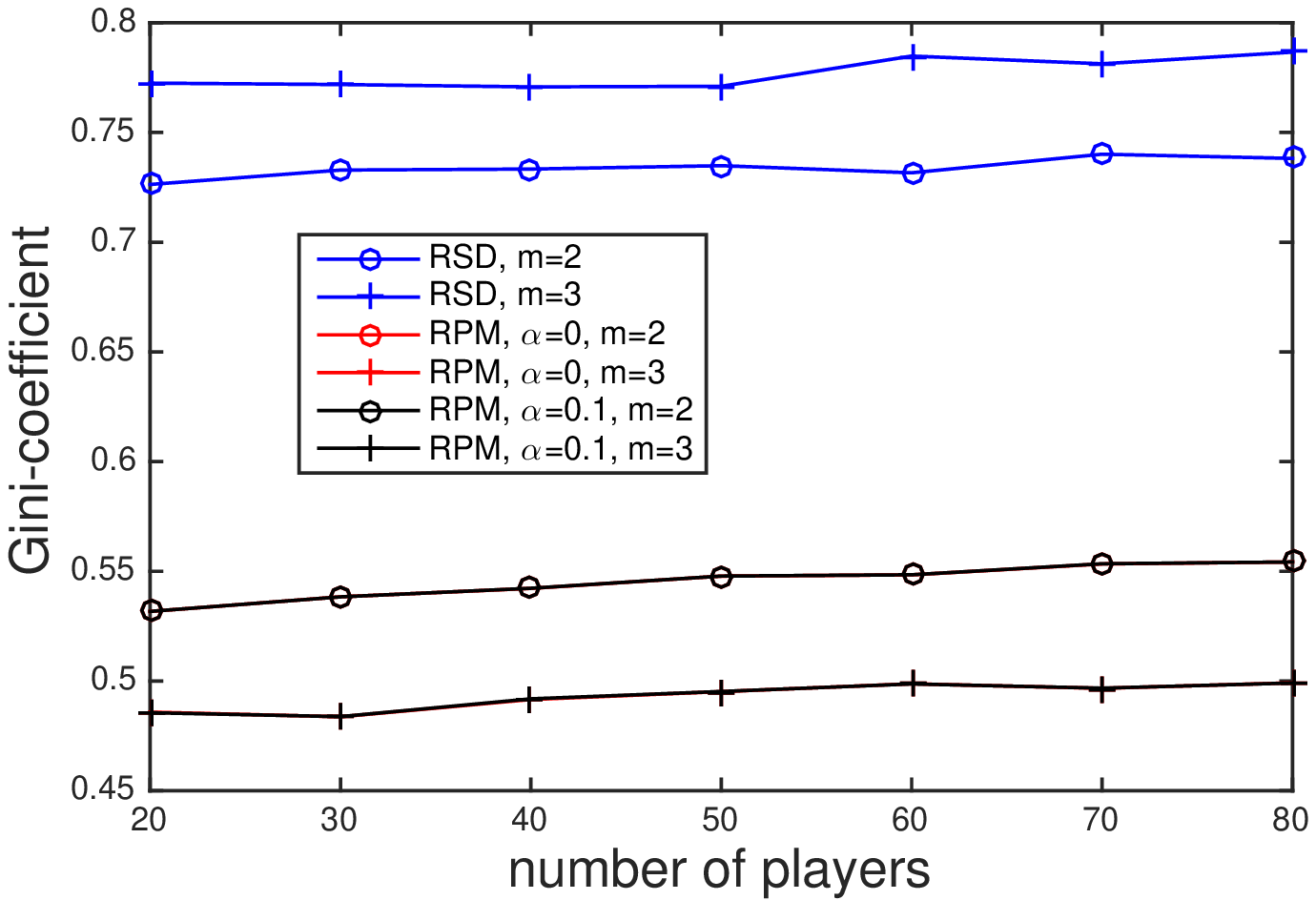} 
		\caption{scale-free network}
		\label{fig:Gini_Scale}
	\end{subfigure}
\begin{subfigure}[b]{0.45\textwidth}
		\includegraphics[width = \textwidth, height=35mm]{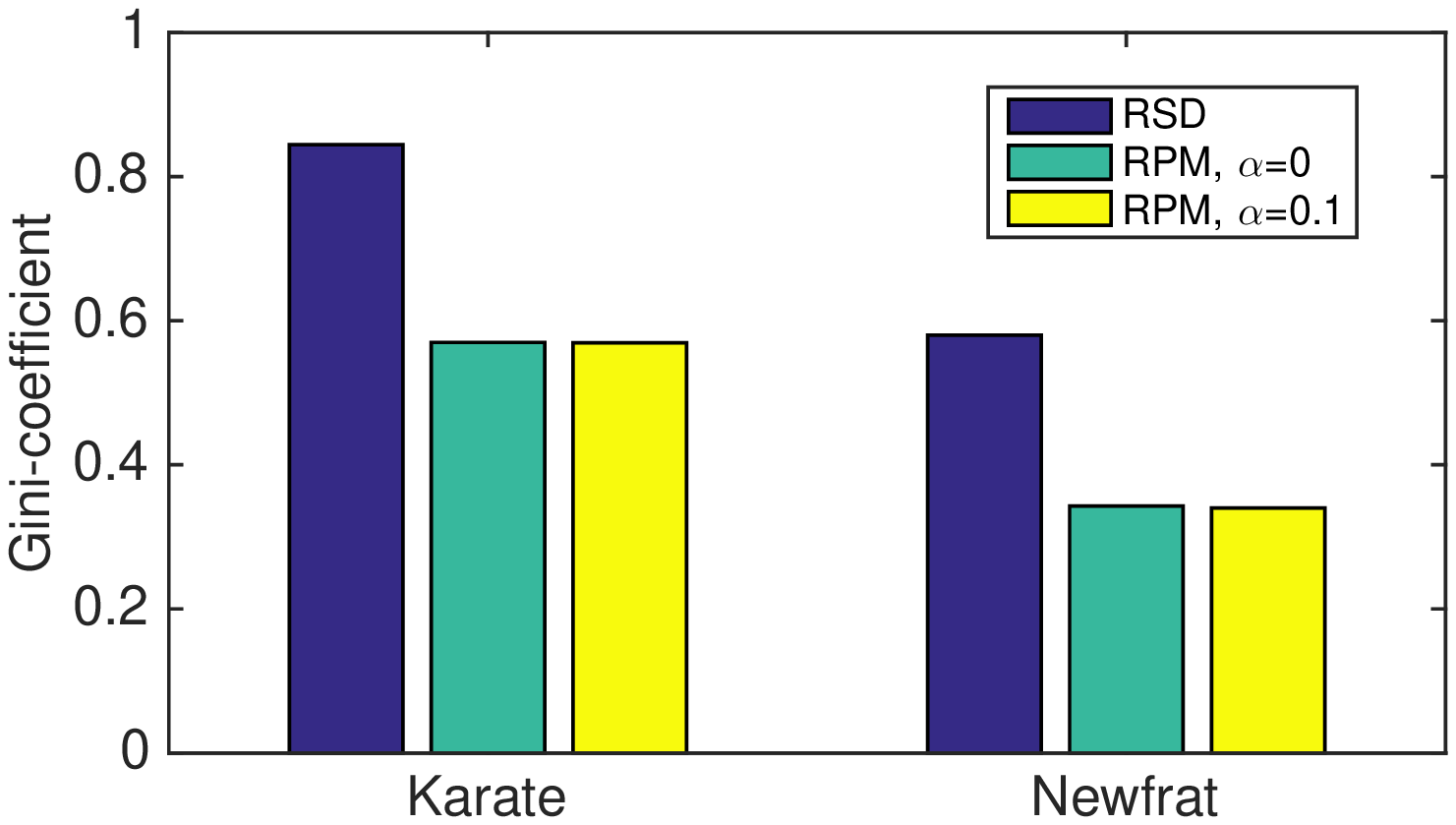} 
		\caption{Karate and Newfrat}
		\label{fig:Gini_Real}
	\end{subfigure}
\caption{Gini coefficient for the roommate problem}
\label{fig:Gini_Team}
\end{figure}

The Gini coefficient measures the inequality of player's utilities. It is
extracted based on the Lorenz curve~\citep{lorenz1905methods}.\footnote{%
The proportion of the total utility of the players that is cumulatively
earned by the bottom $x\%$ of the population.} A Gini coefficient of zero
expresses perfect equality, where all the players have the same utility,
while a Gini coefficient of one expresses maximal inequality among values
(e.g., for a large number of players, where one team is composed of
soulmates and all the players are matched to their least preferred team).

Correlation between utility and rank considers each random ranking of
players in $O$ used for both RSD and RPM, along with corresponding utilities 
$u_i(\pi)$ of players for the partition $\pi$ generated by the mechanism,
and computes the correlation between these. It thereby captures the relative
advantage that someone has by being earlier (or later) in the order to
propose than others, and is a key cause of ex post inequity in RSD. We view
the correlation measure as perhaps the most meaningful criterion of fairness
for mechanisms based on random player rankings: for example, someone who is
extremely unpopular is likely to have lower utility than others, but that's
likely to remain the case for any team formation mechanism with good
efficiency properties. On the other hand, this may be relatively invariant
of the ex post position that the player has in the order of proposers.

\begin{figure}[hbtp]
\centering
\begin{subfigure}[b]{0.45\textwidth}
		\includegraphics[width = \textwidth, height=32mm]{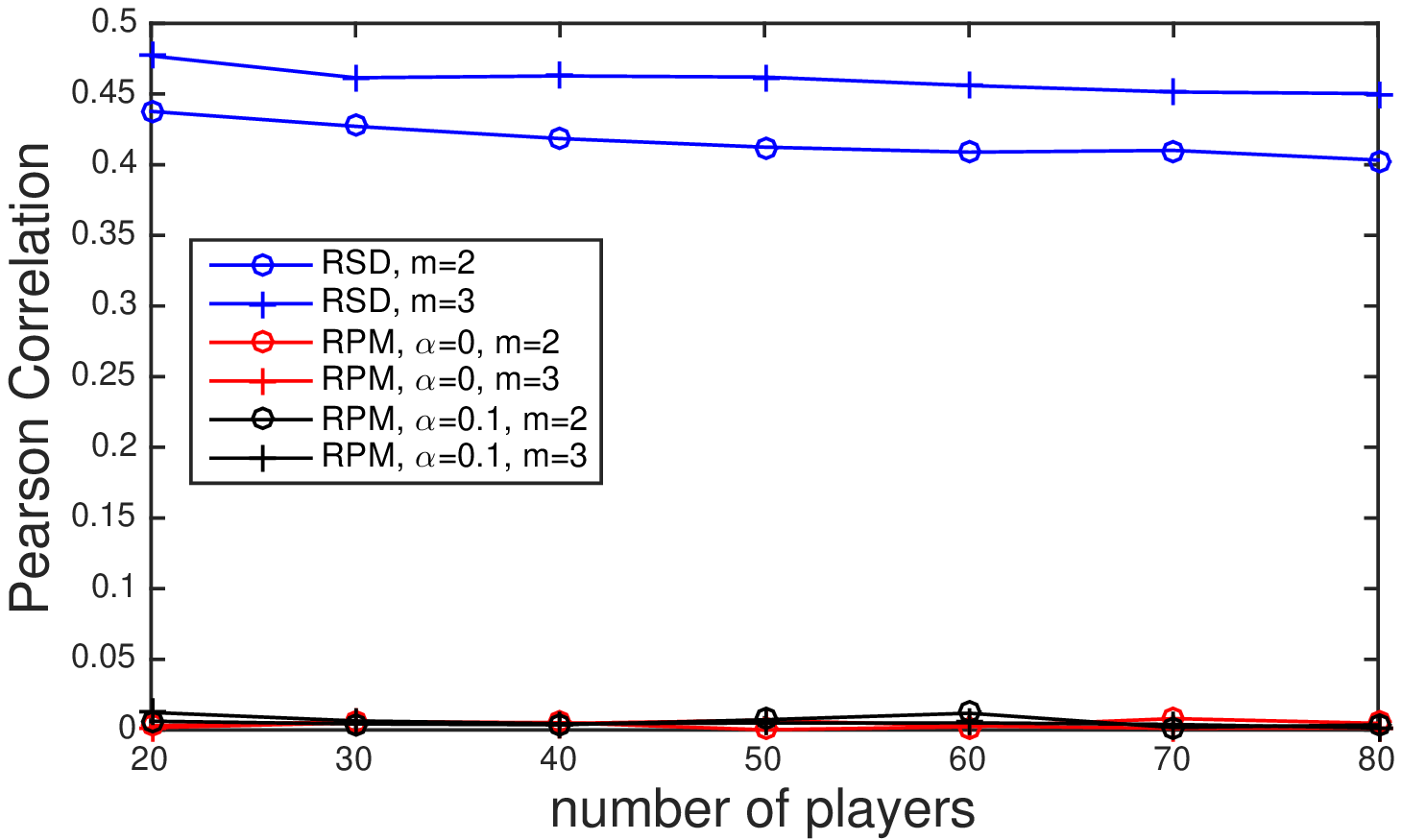} 
		\caption{scale-free network}
		\label{fig:Correlation_Scale}
	\end{subfigure}
\begin{subfigure}[b]{0.45\textwidth}
		\includegraphics[width = \textwidth, height=32mm]{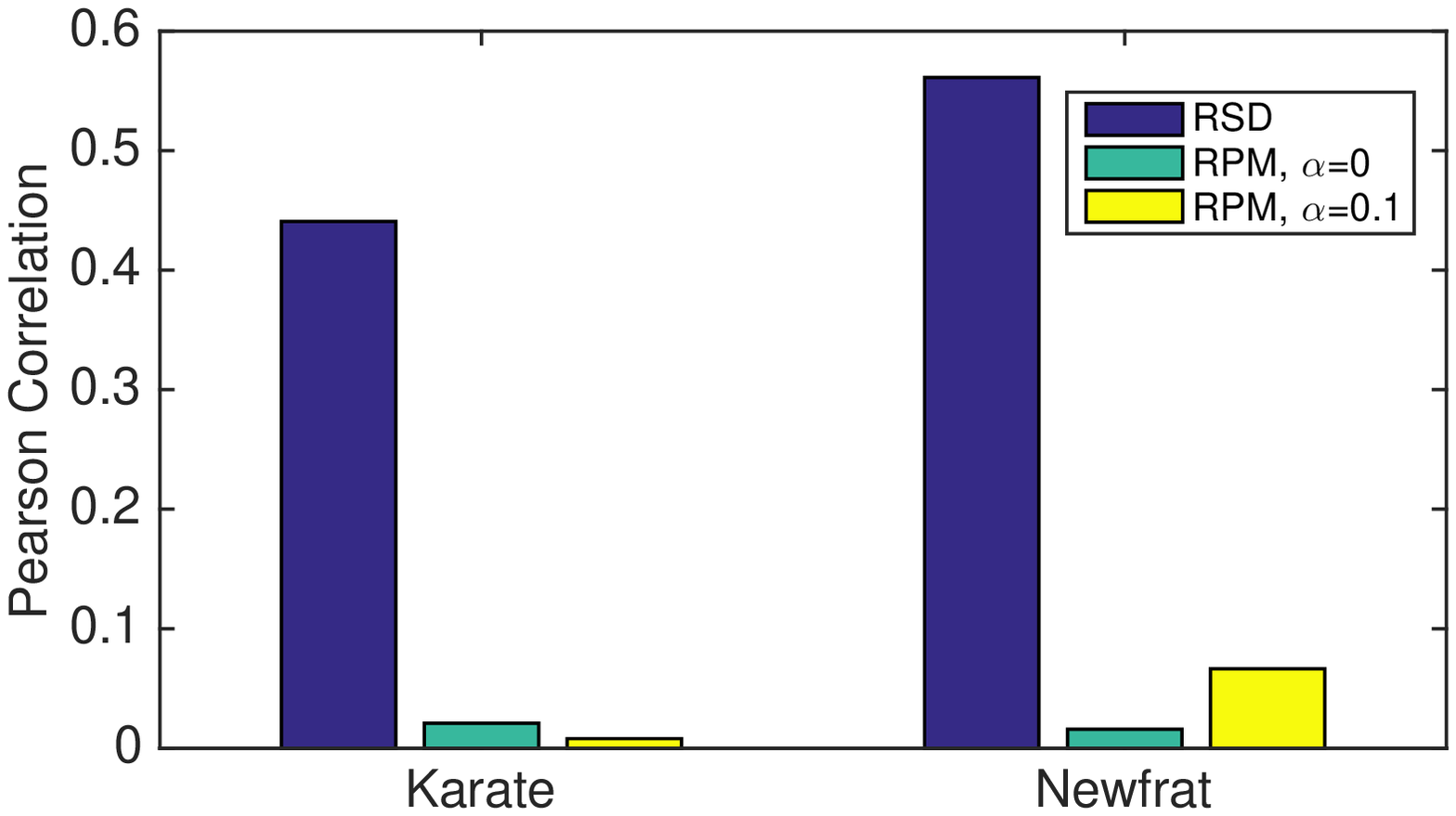} 
		\caption{Karate and Newfrat}
		\label{fig:Correlation_Real}
	\end{subfigure}
\caption{Pearson Correlation for the roommate problem}
\end{figure}

\begin{figure}[h]
\centering
\begin{subfigure}[b]{0.45\textwidth}
		\includegraphics[width=\textwidth, height=32mm]{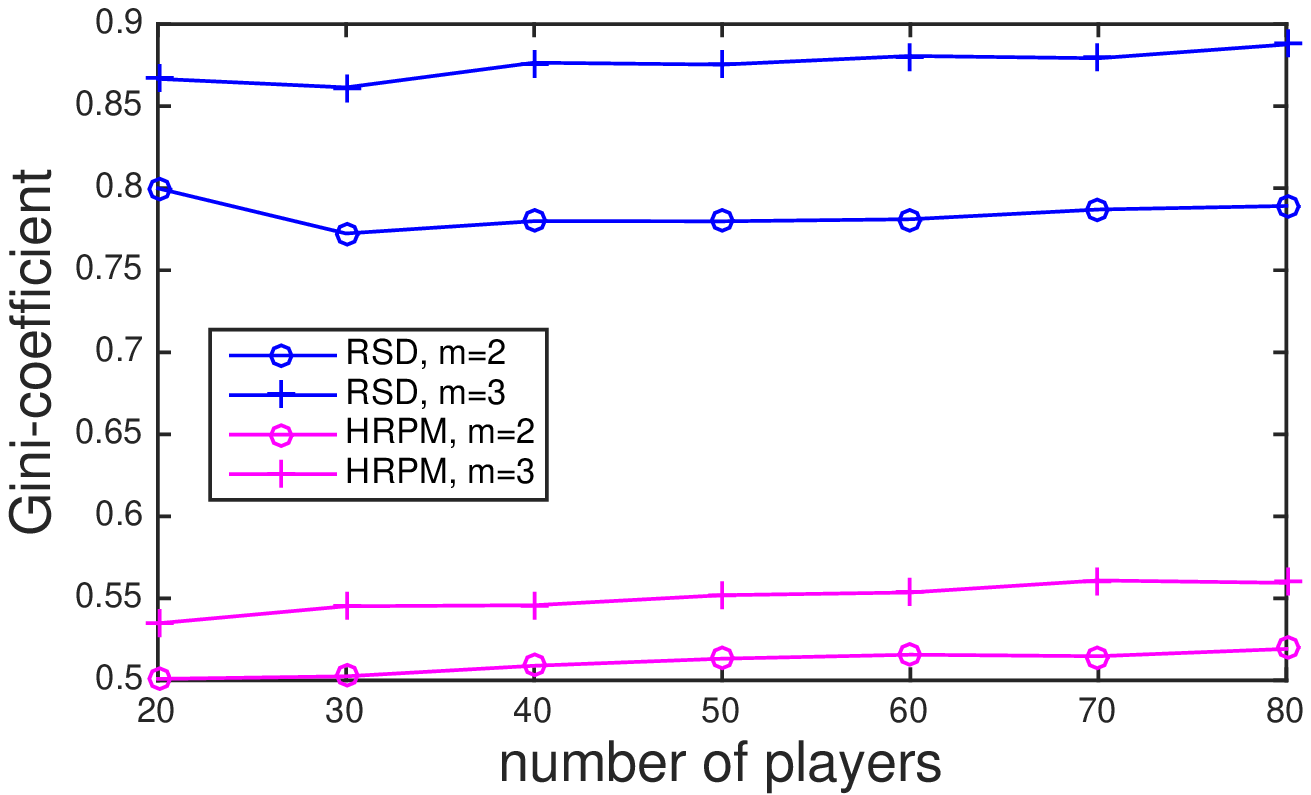} 
		\caption{scale-free networks}
		\label{fig:Gini_Scale_three}
	\end{subfigure}
\begin{subfigure}[b]{0.45\textwidth}
		\includegraphics[width=\textwidth, height=32mm]{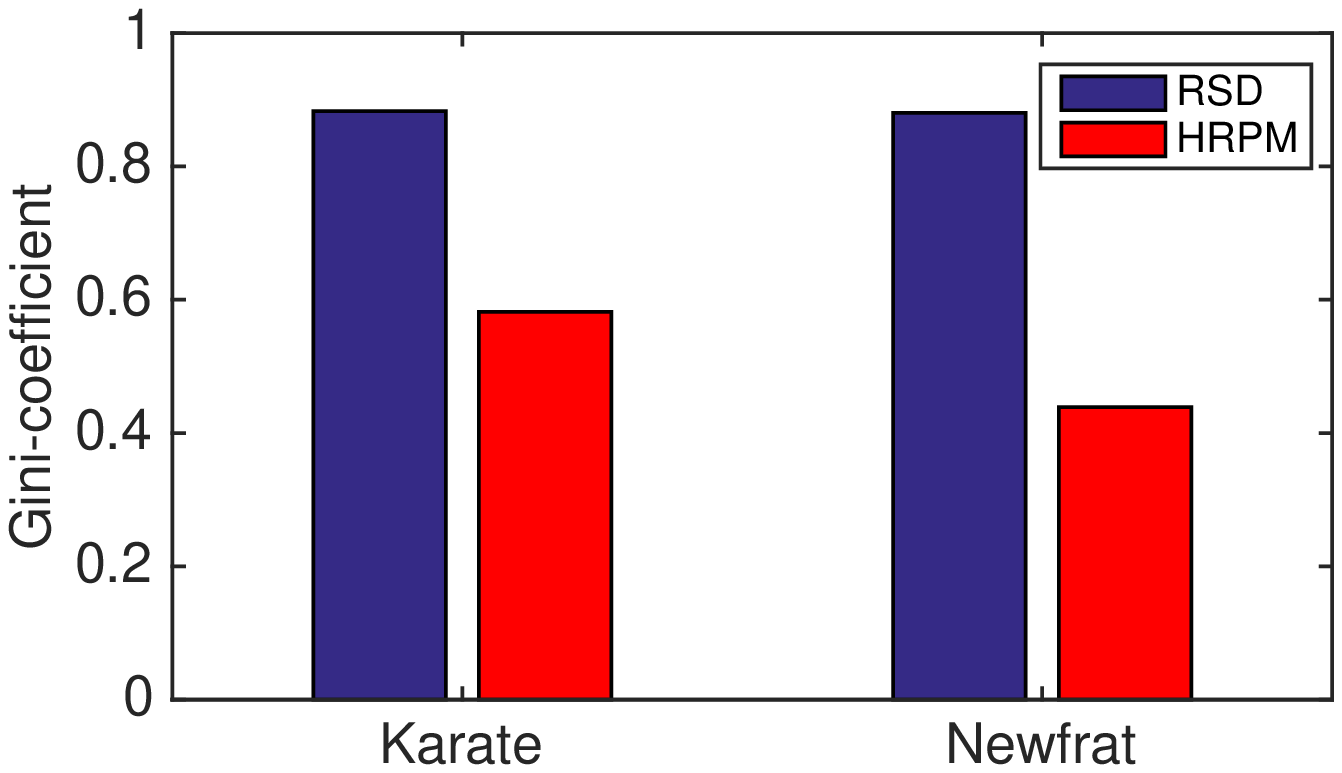} 
		\caption{Karate and Newfrat}
		\label{fig:Gini_Real_three}
	\end{subfigure}
\caption{Gini coefficient for the trio-roommate problem}
\label{fig:Gini_three}
\end{figure}

\begin{figure}[h]
\centering
\begin{subfigure}[b]{0.45\textwidth}
		\includegraphics[width=\textwidth, height=32mm]{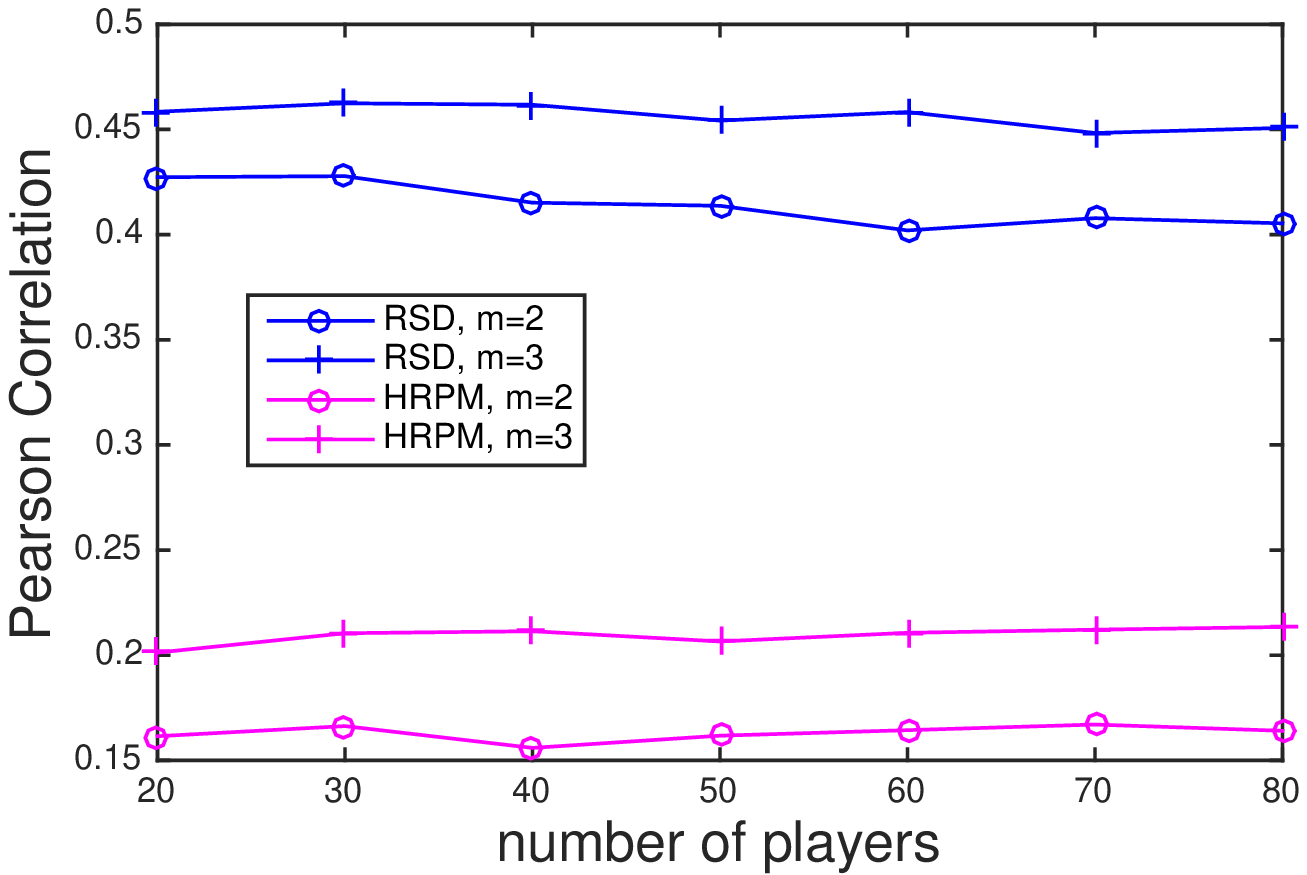} 
		\caption{scale-free networks}
		\label{fig:Correlation_Scale_three}
	\end{subfigure}
\begin{subfigure}[b]{0.45\textwidth}
		\includegraphics[width=\textwidth, height=32mm]{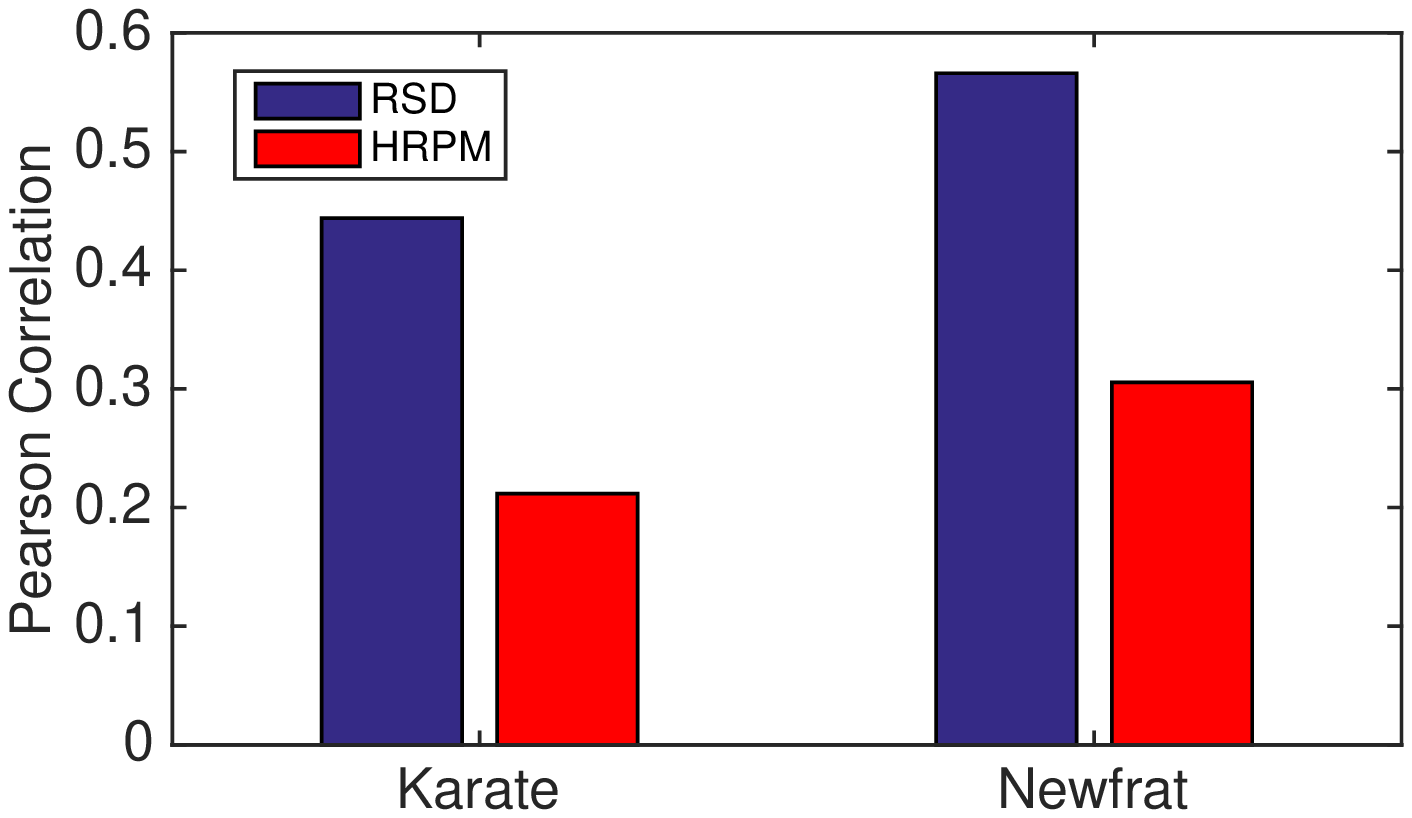} 
		\caption{Karate and Newfrat}
		\label{fig:Correlation_Real_three}
	\end{subfigure}
\caption{Pearson Correlation for the trio-roommate problem}
\label{fig:Correlation_three}
\end{figure}

Our experiments on the roommate problem show that RPM is significantly more
equitable than RSD on scale-free networks (Figures \ref{fig:Gini_Scale} and~%
\ref{fig:Correlation_Scale}), as well as on the Karate club network and
Newfrat dataset (Figures~\ref{fig:Gini_Real} and~\ref{fig:Correlation_Real}%
). The differences between exact and approximate RPM are negligible in most
instances.

In the trio-roommate problem, HRPM ($\beta=0.6$) is much more equitable than
RSD as shown in Figures \ref{fig:Gini_three} and \ref{fig:Correlation_three}%
. These results are statistically significant ($p<0.01$).

\bibliographystyle{plain}
\bibliography{rpm}

\end{document}